\documentclass[12pt,reqno]{article}

\usepackage[english]{babel}
\usepackage{amsfonts, amssymb, amsmath, amsthm} 

\usepackage{eucal}
\usepackage{nicefrac}
\usepackage{color}
\usepackage{mathrsfs}
\usepackage{graphicx}

\usepackage[letterpaper,top=2cm,bottom=2cm,left=3cm,right=3cm,marginparwidth=1.75cm]{geometry}

\usepackage[utf8]{inputenc}

\usepackage{mathtools}

\parskip=\medskipamount
\newtheorem{definition}{\noindent \noindent {\bf
Definition}}[section]
\newtheorem{lem}{{\bf Lemma}}[section]
\newtheorem{theorem}{Theorem}[section]
\newtheorem{prop}{{\bf Proposition}}[section]

\usepackage{hyperref}
\hypersetup{colorlinks=true, allcolors=blue}

\title{Coisotropic embeddings of precosymplectic manifolds}
\author{\sffamily
Manuel de León$^{1,2}$, \thanks{mdeleon@icmat.es\qquad\qquad ORCID: 0000-0002-8028-2348}\, 
Pablo Soto Martín$^{1, 3}$, \thanks{pablo.soto@estudiante.uam.es\qquad\qquad ORCID: 0009-0001-6958-3182}\,
\\[1ex]
\\[0.1ex]
\normalsize\itshape\sffamily
$^1$Instituto de Ciencias Matemáticas (CSIC), Madrid, Spain
\\[0.1ex]
\normalsize\itshape\sffamily
$^2$Real Academia de Ciencias, Madrid, Spain
\\[0.1ex]
\normalsize\itshape\sffamily
$^3$Universidad Autónoma de Madrid, Spain
\\[0.1ex]}

\begin{document}
\maketitle

\begin{abstract}
In this paper we provide a complete characterisation of coisotropic embeddings of precosymplectic manifolds into cosymplectic manifolds. This result extends a theorem of Gotay about coisotropic embeddings of presymplectic manifolds. We also extend to the cosymplectic case some results of A. Weinstein which generalise the Darboux theorem. While symplectic geometry is the natural framework for developing Hamiltonian mechanics, cosymplectic geometry is the corresponding framework for time-dependent Hamiltonian mechanics. The motivation behind proving this theorem is to generalise known results for symplectic geometry to cosymplectic geometry, so that they can be used to study time-dependent systems, for instance for the regularization problem of singular Lagrangian systems. 
\end{abstract}

\section{Introduction}

Symplectic geometry is the natural framework for developing Hamiltonian mechanics. Indeed, the phase space of a classical Hamiltonian system is the cotangent bundle $T^*Q$ of the configuration manifold $Q$, which is endowed with a canonical symplectic structure $\omega_Q$. The solutions of the Hamilton equations are then the integral curves of the Hamiltonian vector field $X_H$ obtained from the Hamiltonian energy $H$ using the symplectic form $\omega_Q$. In the case of Lagrangian mechanics, the velocity space is the tangent bundle $TQ$ of the configuration manifold. In this case there is no canonical symplectic structure, but it is given using the Lagrangian function $L$ and the so-called almost-tangent geometry of the tangent bundle; the 2-form is usually denoted by $\omega_L$. In the latter case, if the Lagrangian is regular (which is usual for Lagrangians appearing in mechanics) the corresponding Hamiltonian vector field is a Second Order Differential Equation (SODE) whose solutions are just the ones of the Euler-Lagrange equations determined by the Lagrangian (see \cite{Abraham1967Foundations,deLeon2011methods}).

This geometrical description of the mechanics is valid when the Hamiltonian or the Lagrangian does not explicitly depend on time; if it does, we must consider the spaces $\mathbb{R} \times T^*Q$ and $\mathbb{R} \times TQ$, respectively. Now, the geometric structures we must use to obtain the dynamics are the so-called cosymplectic structures \cite{albert,cantrijn, deLeon2011methods,cosimpl}. A cosymplectic structure on a manifold $M$ is a pair $(\Omega, \eta)$ where $\Omega$ is a closed 2-form and $\eta$ a closed 1-form, such that $\eta \wedge \Omega^n \not= 0$, where the dimension of $M$ is precisely $2n+1$.

Just as symplectic geometry has allowed us to obtain new results in mechanics, such as symplectic reduction, the various generalizations of Noether's theorem, study of stability, development of geometric integrators, among many others; something similar occurs in the case of cosymplectic geometry, although the extensions are in many cases not as direct as they might seem at first glance.

In this study we have focused on the extension of a result of great interest in symplectic geometry, the so-called coisotropic embedding theorem. If $N$ is a submanifold of a symplectic manifold $(M, \omega)$, we say that it is coisotropoic if $TN^\perp \subset TN$, where $TN^\perp$ is the symplectic complement of $TN$ in the tangent bundle $TM$. On the other hand, a presymplectic manifold is a pair $(P, \omega)$ where $\omega$ is closed but not of maximal rank (if the rank is maximal, then we are in presence of a symplectic manifold). 

Mark J. Gotay \cite{gotay1982coisotropic} has proved a relevant result that says that any presymplectic manifold can be coisotropically embedded in a symplectic manifold (see also \cite{marle83sous,zambon2014equivalences}). The existence is unique up to symplectomorphisms (this result lies on previous ones by Alan Weinstein \cite{weinstein1977lectures, weinstein71symplectic}).
One of the interests of this result is the application to the so-called problem of regularization. Indeed, if a Lagrangian function is not regular (this happens in several physical theories), the corresponding 2-form $\omega_L$ is presymplectic (assuming some regularity condition on the rank). For this kind of Lagrangian functions, P.A.M. Dirac developed a constraint algorithm (the so-called Dirac-Bergmann algorithm \cite{Dirac_1950}) that has been geometrized by Mark J. Gotay, James N. Nester and George Hinds \cite{gotay1,GotaySingularLagrangians1,GotaySingularLagrangians2}. The algorithm produces a sequence of constraint submanifolds just to find a final constraint submanifold. At this stage one can use the coisotropic embedding theorem to embed our system in a larger one that it is symplectic and the corresponding dynamics are conveniently related (see Alberto Ibort and Jes\'us Mar{\'i}n-Solano \cite{ibort-marin}).

There exists an extension of the Dirac-Bergmann constraint algorithm for singular time-dependent Lagrangian systems, developed by Domingo Chinea, Manuel de Le\'on and Juan Carlos Marrero \cite{chinea}. To be able to achieve a regularization for singular Lagrangians depending on time, we first need to extend the above theorem to the cosymplectic scenario. The objective of this paper is to extend the coisotropic embedding theorem to precosymplectic manifolds. The application of the above results to the case of singular time-dependent Lagrangian systems is the subject of a forthcoming paper in elaboration.

This paper is structured as follows. In Section 2 we develop all the mathematical background on cosymplectic geometry. Section 3 is devoted to the extension of the coisotropic embedding theorem to precosymplectic manifolds. Finally, we mention some conclusions and future work. The coisotropic embedding theorem for presymplectic manifolds is found on Appendix A, in order to to make the paper as self-contained as possible. Furthermore, it brings together in a complete way the details of the demonstrations that are collected in several papers.

\section{Cosymplectic geometry}

Cosymplectic manifolds are the odd-dimension counterpart of symplectic manifolds. The original definition of cosymplectic manifolds comes from Libermann \cite{libermann1959sur} in 1959:

\begin{definition}
    Let $M$ be a manifold of odd dimension $2n + 1$. A \textbf{cosymplectic} manifold is a triple $(M, \Omega, \eta)$, where $\eta$ is a closed 1-form and $\Omega$ is a closed 2-form such that $\eta \wedge \Omega^n$ is a volume form.
\end{definition}

If the 2-form $\Omega$ does not have maximal rank, but has constant rank $2r \leq 2n$ and $\eta \wedge \Omega^r \neq 0$, we say $(M, \Omega, \eta)$ is \textbf{precosymplectic}.

\

We can generalise this definition to vector bundles; indeed, if $\pi :E \longrightarrow M$ is a vector bundle such that there are a 2-form and a 1-form $\Omega$ and $\eta$, respectively, such that any fiber $E_x$ is a cosymplectic vector space with forms $\Omega(x)$ and $\eta(x)$. If both $x \rightarrow \Omega(x)$ and $x \rightarrow \eta(x)$ are smooth, then the vector bundle is said to be cosymplectic.

\

Every cosymplectic structure $(\Omega, \eta)$ on $M$ induces an isomorphism \cite{CAPPELLETTI_MONTANO_2013} $\flat$ from the vector fields in $M$ to the 1-forms in $M$, defined by

\begin{equation}
    \label{cosympl:bemol}
    \flat_{(\Omega, \eta)}(X) = i_X\Omega + \eta(X)\eta,
\end{equation}

\noindent for every vector field $X \in \mathfrak{X}(M)$. We will call $\xi = \flat_{(\Omega, \eta)}^{-1}(\eta)$ the Reeb vector field of the cosymplectic manifold $(M, \Omega, \eta)$. It is also characterised by the equations

\begin{equation}
    \label{cosympl:reeb}
    i_\xi\Omega = 0, \;\; \eta(\xi) = 1.
\end{equation}

This vector field is very useful as it will give us the dynamics of a time-dependent Lagrangian or Hamiltonian system. As with the symplectic case, we also have structure-preserving morphisms for cosymplectic manifolds.

\begin{definition}
    Let $(M, \Omega, \eta), (N, \Omega', \eta')$ be cosymplectic manifolds and $h: M \rightarrow N$ a differentiable mapping. We say that $h$ is cosymplectic if $\Omega = h^*\Omega'$ and $\eta = h^*\eta'$. We will call $h$ a cosymplectomorphism when $h$ is a global diffeomorphism.
\end{definition}

Given a cosymplecticmorphism $h: M \rightarrow N$, the Reeb vector field $\xi$ of $M$ is $h$-projectable and, its projection is the Reeb vector field $\xi'$ of $N$, that is, $$Th \circ \xi = \xi' \circ h.$$

We now define the orthogonal complement:

\begin{definition}
    \label{def:cosymp:orthogonal}
     Let $(M, \Omega, \eta)$ be a cosymplectic manifold. Let $p\in M$ be a point and $K_p$ be a subspace of $T_pM$. Then, we define the cosymplectic orthogonal complement of $K_p$ in $T_pM$ as
     \begin{equation}
        (K_p)^\perp = \{ X\in T_pM | \; \eta(X) = 0, \Omega(X,Y) = 0, \;\forall Y \in K_p\}
    \end{equation}
\end{definition}

We will say that $K \subset M$ is a coisotropic submanifold if $\xi \in K$ and $(TK)^\perp \subset TK$. Notice that $K^\perp \subset \ker \eta$.


There is also a corresponding Darboux Theorem for precosympletic manifolds (see \cite{godbillon}). Any precosymplectic manifold $(M, \Omega, \eta)$ of dimension $2p + k + 1$, with rank $\Omega = 2p$,  admits local coordinates $(x^i, t, z^\mu), \; 1\leq i\leq 2p, 1 \leq \mu \leq k$ such that $$\Omega = \sum_{i=1}^p dx^i \wedge dx^{p + i}, \quad \eta = dt, \quad \xi = \frac{\partial}{\partial t}.$$

\subsection{Coisotropic embeddings}

We are interested in a specific type of embedding: coisotropic embeddings. These embeddings can be defined for precosymplectic manifolds.

\begin{definition}
    \label{presymplectic:coisotropic}
    \cite{gotay1982coisotropic} Let $(M, \Omega,\eta)$ be a precosymplectic manifold and let $(P, \Omega', \eta')$ be a cosymplectic manifold. Then, a coisotropic embedding of $(M, \Omega, \eta)$ into $(P, \Omega', \eta')$ is a closed embedding $j: M \rightarrow P$ such that

    \begin{itemize}
        \item $j^*\Omega' = \Omega$,
        \item $j^*\eta' = \eta$,
        \item $TM^\perp \subset Tj(TM).$
    \end{itemize}
\end{definition}

Following the definition of coisotropic embedding, we introduce the notion of {\it neighbourhood equivalence}, which will help us to classify coisotropic embeddings. Given $(M, \Omega, \eta)$ a precosymplectic manifold and embeddings $j_1 : M \rightarrow (P_1, \Omega_1, \eta_1), j_2: M \rightarrow (P_2, \Omega_2, \eta_2)$, where $(P_1, \Omega_1, \eta_1)$ and $(P_2, \Omega_2, \eta_2)$ are cosymplectic manifolds, we say that they are neighbourhood equivalent if there exist

\begin{enumerate}
    \item open neighbourhoods $U_1 \subset P_1$ of $j_1(M)$ and $U_2 \subset P_2$ of $j_2(M)$,
    \item a cosymplectomorphism $\psi: (U_1, \Omega_1, \eta_1) \rightarrow (U_2, \Omega_2,\eta_2)$ such that $\psi \circ j_1 = j_2$.
\end{enumerate}

\section{Coisotropic embeddings of precosymplecitc manifolds}

In this section we extend the coisotropic embedding theorem to precosymplectic manifolds. The statement of the theorem we want to prove is the following:

\begin{theorem}
    \label{thm:cois:precosymp}
    Every coisotropic submanifold of a cosymplectic manifold has an induced precosymplectic structure. Additionally, every precosymplectic manifold can be coisotropically embedded into a cosymplectic manifold. This embedding is unique up to a local cosymplectomorphism.
\end{theorem}

Theorem \ref{thm:cois:precosymp} can be divided into three smaller results, one corresponding to each statement of the theorem. We will denote these results as direct, existence and local uniqueness theorems respectively. Even though the ideas to proof the theorem for the precosymplectic case are similar to the presymplectic one, they are not as direct as they might seem at first glance.

\subsection{Direct theorem}

\begin{theorem}
    Let $(P, \Omega, \eta)$ be a cosymplectic manifold, and let $M$ be a coisotropic submanifold. Let $j: M \xhookrightarrow{} P$ be the inclusion of $M$ in $P$. Then $(M, j^*\Omega, j^*\eta)$ is a precosymplectic manifold, with $j^*\Omega$ having constant rank 2dim $M$ - dim $P$ - 1.
\end{theorem}

\begin{proof}
    Let $\Omega_M = j^*\Omega$ and $\eta_M = j^*\eta$. Note than both $\Omega_M$ and $\eta_M$ are closed as $d(j^*\Omega) = j^*d\Omega = 0$ and $d(j^*\eta) = j^*d\eta = 0$.

    Let $x$ be a point of $M$. Then, rank $\Omega_M(x)$ = dim (Im $\flat_{\Omega_M(x)}$) = dim $T_xM$ - dim (ker $\flat_{\Omega_M(x)}$), where $\flat_{\Omega_M(x)} = i_X \Omega_M(x)$. To avoid a notation overlap, for this proof we denote as $(T_xM)^\perp_{symp}$ the symplectic complement with respect to $P$ (that means, on the presymplectic vector bundle $(TP, \Omega)$), and we denote as $(T_xM)^\perp_{cos}$ the cosymplectic complement with respect to $P$ given in definition \ref{def:cosymp:orthogonal}.

    Then, considering the linear mapping $\flat_{\Omega_M(x)}: T_xM \rightarrow (T_xM)^*$, we have 
    $$
    {\text{ker } \flat_{\Omega_M(x)} = T_xM\; \cap \;(T_xM)^\perp_{symp}}
    $$
    and
    $$(
    T_xM)^\perp_{symp} = (T_xM)^\perp_{cos} \oplus \text{ker } \Omega (x).
    $$
    Since $M$ is coisotropic, we know that $\text{ker } \Omega (x)\subset T_xM$ and that $(T_xM)^\perp_{cos} \subset T_xM$. Therefore, $\text{dim (ker } \flat_{\Omega_M(x)}) = \text{dim }(T_xM)^\perp_{symp}$. By the properties of the symplectic complement (\cite{deLeon2011methods}), we have $$\text{dim } (T_xM)^\perp_{symp} = \text{dim } T_xP + \text{ dim (ker } \Omega (x)\; \cap \; T_xM) - \text{ dim } T_xM = \text{dim } P - \text{ dim } M + 1.$$ Then,
    \begin{align*}
        \text{rank }\Omega_M(x) &= \text{dim (Im }\flat_{\Omega_M (x)}) = \text{dim }T_xM - \text{dim (ker }\flat_{\Omega_M(x)}) = 2 \text{dim }T_xM - \text{dim }T_xP \\
        &= 2 \text{dim }M - \text{dim }P - 1. &&
    \end{align*} 
    Let $\xi_x$ be the Reeb vector field in $T_xP$. Since $\xi_x \in T_xM$, we have $i_\xi \Omega_M(x) = 0$ and $i_\xi \eta_M = 1$. This implies that $\xi_x$ is a Reeb vector field of $(M, \Omega_M, \eta_M)$ and that $\Omega_M(x)^r \wedge \eta_M \neq 0$, where $r = \text{dim }M - \frac{\text{dim }P + 1}{2}$.

    
\end{proof}

\subsection{Existence theorem}

The bundle in which we will embed $M$ for the existence theorem is similar to the one used for the presymplectic case. Instead of using the characteristic bundle $K=\text{ ker } \Omega$, we will use $K = \text{ ker } \Omega \cap \text{ ker } \eta$.

We define $K^*$ as the vector bundle whose fibres are $$(K^*)_x = (K_x)^*,$$ the dual of the fibres of $K$.

\begin{theorem}
\label{thm:cosympl_cois:existence}
    Let $(M, \Omega, \eta)$ be a precosymplectic manifold of dimension $2p + k + 1$, where $\Omega$ has constant rank $2p$. There exists a cosymplectic structure on a neighborhood of the zero-section of $K^*$, where $K = $ ker $\Omega \; \cap$ ker $\eta$. Moreover, $(M, \Omega, \eta)$ may be coisotropically embedded in this neighborhood as the zero-section.

\end{theorem}

\begin{proof}
    As $\Omega$ is precosymplectic, it is possible to find local Darboux coordinates in $M$, $(x^i, t, z^\mu), \; i=1,...,2p, \; \mu = 1,...,k$ such that $\Omega = dx^i \wedge dx^{p+i}$ and $\eta = dt$. The coordinates $z^\mu$ are local coordinates along the fibres of $K$, and the vector fields $\partial/\partial z^\mu$ form a local basis of $\Gamma(K)$. Let $\{(U_\alpha, \varphi_\alpha)\}$ be an atlas of $M$ corresponding to the local Darboux coordinates. Then, there exists a partition of unity $\{(U_i, f_i)\}$ subordinate to $\{U_\alpha\}$ (\cite{deLeon2011methods}). Since each $U_i$ is contained in some $U_\alpha$, if we set $\varphi_i = \varphi_\alpha|_{U_i}$, then $\{(U_i, \varphi_i)\}$ is also an atlas. We define the vector field $\xi$ as 

    $$(\xi_i)_x = 
    \begin{cases} 
    f_i(x) \frac{\partial}{\partial t} & \text{if } x \in U_i \\
    0 & \text{if }  x \notin U_i
    \end{cases}$$

    $\xi = \sum_i \xi_i$ is a vector field over $M$ such that $\xi_x = \frac{\partial}{\partial t}$ in the Darboux coordinates chosen above.
    
    Let $K = $ ker $\Omega \; \cap$ ker $\eta$. Let $G$ be a complement of $K$ such that $\xi_x \in G_x$ for all $x \in M$. This is always possible because it never occurs that $\xi_x = ( \frac{\partial}{\partial t} )_x \in K_x$. Then, $TM = K \oplus G$. $T_MK^*$ has the canonical decomposition $T_MK^* = TM \oplus K^* = G \oplus K \oplus K^*$. Let $p_G: TM \rightarrow K$ be the natural projection such that, if $v = (v_k, v_g) \in TM$, with $v_k \in K$ and $v_g\in G$, then $p_G(v) = p_G((v_k, v_g)) = v_k$. Let $j_G$ be the adjoint of $p_G$, i.e. if $\alpha \in K^*, \; j_G(\alpha) = \alpha \circ p_G$.

    Let $j: K^* \rightarrow T^*M$ be the inclusion. Extending the local Darboux coordinates in $M$, for any element of $K^*$, there exist local coordinates $(x^i, t, z^\mu, b_\mu)$, such that $b_r(\frac{\partial}{\partial z^s}) = \delta_{r, s}$. Then, $\pi^*\Omega = dx^i\wedge dx^{p+i}$ and $\pi^*\eta = dt$ , where $\pi$ is the canonical projection $\pi: K^* \rightarrow M$.

    Any vector $X \in TK^*$ can be expressed in the form 
    $$
    X = \sum_{i=1}^{2p} X^i \frac{\partial}{\partial x^i} + T\frac{\partial}{\partial t} + \sum_{r=1}^k Z^r \frac{\partial}{\partial z^r} + \sum_{s=1}^k B_s \frac{\partial}{\partial b_s},
    $$ 
    where $X^i, T, Z^r, B_s$ are functions of the local coordinates $(x^i, t, z^\mu, b_\mu)$. We have that

    \begin{align*}
        T\pi(X) &= \sum_{i=1}^{2p} X^i \frac{\partial}{\partial x^i} + T \frac{\partial}{\partial t} + \sum_{r=1}^k Z^r \frac{\partial}{\partial z^r},\\
        p_G \circ T\pi(X) &= \sum_{r=1}^k \Big(Z^r + \sum_{i=1}^{2p} A^r_i X^i + C^rT\Big)\frac{\partial}{\partial z^r},
    \end{align*}

    \noindent where $A^r_i, C^r$ depend on $(x^i, t, z^\mu)$ and the projection $p_G$. By definition of $G$, $\frac{\partial}{\partial t} \in G$. Then, the projection of $\frac{\partial}{\partial t}$ by $p_G$ is 0, so $C^r = 0$. Let $\lambda_{p_G}$ be the generalised Liouville form \cite{marle83sous}. Then, 
    $$
    \lambda_{p_G}(X) = \sum_{r=1}^k B_r\Big(Z^r + \sum_{i=1}^{2p} A^r_i X^i\Big).
    $$ 
    $\lambda_{p_G}(X)$ has the same expression than in the proof of theorem \ref{thm:sympl_cois:existence}. Let $\Omega_G = \pi^*\omega - \omega_{p_G}$ be a 2-form on $K^*$. Then, $\Omega_G$ is locally given by equation \eqref{sympl_cois:existence:omega}. $\Omega_G$ is smooth and closed. In the zero section of $K^*$, $b_r = 0$ and we can see that $\Omega_G$ has constant rank $2p + 2k$. So, there exists an neighbourhood of the zero section of $K^*$ such that $\Omega_G$ has constant rank $2p + 2k$.

    Let $\tilde{\eta}$ be the 1-form on $K^*$ given by $\tilde{\eta} = \pi^*\eta$. In local Darboux coordinates $\tilde{\eta} = dt$. We have that $i_{\frac{\partial}{\partial t}}\Omega_G = 0$ by equation \eqref{sympl_cois:existence:omega} and that $\tilde{\eta}(\frac{\partial}{\partial t}) = 1$. As $\Omega_G$ has constant rank $2p + 2k$, the pair $(\Omega_G, \tilde{\eta})$ is cosymplectic over a neighbourhood of $M$ in $K^*$ with Reeb vector field locally given by $\xi = \frac{\partial}{\partial t} \in TM$.

    To see that $M$ is a coisotropic submanifold of $K^*$ , it remains to show that $TM^\perp \subset TM$ in $T_MK^*$. We can do a similar argument as with the symplectic case.  Let's suppose it is not true. Then there exists $x \in M$, $Y \in T_xK^*, Y \not\in T_xM$, such that $\tilde{\eta}(Y) = 0$ and $\Omega_G(X,Y) = 0, \forall X \in T_xM$. If we use Darboux local coordinates, then we get that at least there is some $B_\mu \neq 0$. Then, if we pick $X = \partial/\partial z^\mu$, then $\Omega_G(X,Y) = 1$. 
\end{proof}

\subsection{Local Uniqueness theorem}

To proof the Local Uniqueness theorem we will need to extend to cosymplectic manifolds some of the preliminary results we used in the proof of the uniqueness theorem for the symplectic case. We will extend the results proved by Weinstein (Theorem \ref{sympl:extension:same} and Theorem \ref{thm:sympl:ext} in the Appendix). In order to extend Theorem \ref{sympl:extension:same}, we will use two lemmas:

\begin{lem}
    \label{lemma:cosympl:ext:1}
    Let $(\Omega_0, \eta_0)$, $(\Omega_1, \eta_1)$ be cosymplectic structures on a neighbourhood of $P$, a manifold of dimension $2n + 1$, such that they coincide on $M$. Then, there exists a cosymplectic structure $(\bar{\Omega}_0, \eta_1)$ on a neighbourhood of $P$ such that $(P,\Omega_0, \eta_0)$ and $(P, \bar{\Omega}_0, \eta_1)$ are neighbourhood equivalent, and such that the Reeb vector field $\bar{\xi}_0$ of $(\bar{\Omega}_0, \eta_1)$ is the same as the Reeb vector field $\xi_1$ of $(\Omega_1, \eta_1)$.
\end{lem}

\begin{proof}
    Let $\nu = \eta_1 - \eta_0$, and let $\eta_t = \eta_0 + t\nu$, for $t$ between 0 and 1. As $\eta_0, \eta_1$ are closed, $d\nu = 0$ and $\nu |_M = 0$. Then, $d\eta_t = 0$ and $\eta_t |_M = \eta_0 |_M$. As $\eta_t$ is non-degenerate over $M$, there is a neighbourhood of $M$ in $P$ such that $\eta_t$ is non-degenerate. We can use the relative Poincaré lemma to get a function $\varphi$ such that $\nu = d\varphi$ and $\varphi|_M = 0$. Let $N_t$ be a smooth time-dependent vector field such that $\eta_t(N_t) = 1$. Let $Z_t = -\varphi N_t$. Then, $\eta_t(Z_t) = -\varphi$. Since $\varphi |_M = 0$, then $Z_t|_M = 0$ and the time-dependent vector field $\{Z_t\}_{t \in [0,1]}$ is integrable to a one-parameter family $\{g_t\}_{t \in [0,1]}$ of diffeomorphisms of neighbourhoods of $M$ in $P$ (lemma \ref{lemma:integrable}). Then, we have that

    \begin{align*}
        \frac{d}{ds}(g_s^*\eta_s)_{s=t} &= g_t^*\Big( \frac{d\eta_s}{ds} \Big)_{s=t} + g_t^*(d(\eta_t(Z_t))) \\
        &= g_t^*(\nu + d(-\varphi)) \\
        &= 0.
    \end{align*}
    
    Let $g = g_1$. Then $g^*\eta_1 = g_1^*\eta_1 = g_0^*\eta_0 = \eta_0$. Let $\bar{\Omega}_0 = (g^{-1})^*\Omega_0$. As $g|_M$ is the identity, $\bar{\Omega}_0|_M = \Omega_0|_M = \Omega_1|_M$. $\bar{\Omega}_0$ is a 2-form of constant rank $2n$ as $g$ is a diffeomorphism. As $Z_t = 0$ on $M$, the flow in $M$ is the identity, and we have that $g|_M$ and $Tg|_M$ are the identity mappings. It remains to check that $(\bar{\Omega}_0, \eta_1)$ is cosymplectic and that the Reeb vector fields $\bar{\xi}_0$ and $\xi_1$ are the same. To prove it, we will use the flexibility we have in picking the vector field $N_t$.

    Let $\xi_t = a(t)\xi_0 + b(t)\xi_1$, where $a,b$ are smooth functions such that $a(0) = 1,\; a(1)= 0, \; b(0) = 0, \; b(1) = 1$. We define $N_t$ as $\frac{\xi_t}{\eta_t(\xi_t)}$. Then, $N_0 = \xi_0$ and $N_1 = \xi_1$. As $g_t$ is the time-dependent flow with infinitesimal generator $Z_t$ and $N_t$ is parallel to $Z_t$, then the image of $\xi_0 = N_0$ by $dg$ is parallel to $N_1 = \xi_1$. Moreover, $\eta_1(dg(\xi_0)) = g^*\eta_1(\xi_0) = \eta_0(\xi_0) = 1$. Then, we must have that $dg(\xi_0) = \xi_1$. Conversely, we have that $(dg^{-1})\xi_1 = \xi_0$. So, we have that $i_{\xi_1}\bar{\Omega}_0 = i_{\xi_o}\Omega_0 \circ (dg^{-1}) = 0$. Then $(\bar{\Omega}_0, \eta_1)$ is cosymplectic in a neighbourhood of $M$ in $P$ with Reeb vector field $\bar{\xi}_0 = \xi_1$.

    We check now that $N_t$ is well-defined, i.e. $\eta_t(\xi_t) \neq 0$. Indeed, we have
     \begin{align*}
        \eta_t(\xi_t) &= a(t)(1-t) + b(t)(1-t)\eta_0(\xi_1) + b(t)t + a(t)t\eta_1(\xi_0) \\
        &= a(t)(1 - t + t \eta_1(\xi_0)) + b(t)(t + (1-t)\eta_0(\xi_1))
    \end{align*}

    If we define $a(t) = (1 - t + t\eta_1(\xi_0))(1-t)$ and $b(t) = (t + (1 - t)\eta_0(\xi_1))t$, the boundary conditions are respected and we have that $\eta_t(\xi_t) > 0$ as long as it does not happen that
    \[
    \begin{cases}
        1 - t + t\eta_1(\xi_0) = 0 \\
        t + (1 - t)\eta_0(\xi_1) = 0
    \end{cases}
    \]

    If we pick the first equation, we have that $t = \frac{1}{1 - \eta_1(\xi_0)}$. As $\eta_1|_M(\xi_0) = 1$, then there exists a neighbourhood of $M$ in $P$ such that $\left| t \right| = \left| \frac{1}{1 - \eta_1(\xi_0)}\right| > 1$, and then there is no solution to the equation. In this neighbourhood we will have that $N_t$ is well defined and that $(\bar{\Omega}_0, \eta_1)$ is cosymplectic.

\end{proof}

\begin{lem}
    \label{lemma:cosympl:ext:2}
    Let $(\Omega_0, \eta)$ and $(\Omega_1, \eta)$ be cosymplectic structures on a neighbourhood of $M$ in $P$, a manifold of dimension $2n + 1$. Suppose the restrictions of both structures to $T_MP$ are equal and their Reeb vector fields coincide. Then there are neighbourhoods $U_1$ and $U_2$ of $M$ in $P$ and a cosymplectomorphism $f$ from $(U_0, \Omega_0, \eta)$ to $(U_1, \Omega_1, \eta)$ such that $f|_M$ and $Tf|_{T_MP}$ are the identity mappings.
\end{lem}

\begin{proof}
    Let $\omega = \Omega_1 - \Omega_0$ and let $\Omega_t = \Omega_0 + t\omega$, for $t$ between 0 and 1. As $\Omega_0, \Omega_1$ are closed, $d\omega = 0$ and $\omega |_M = 0$. Then, $d\Omega_t = 0$ and $\Omega_t |_M = \Omega_0 |_M$. As $\Omega_t$ has constant rank $2n$ over $M$, there is a neighbourhood of $M$ in $P$ such that $\Omega_t$ has constant rank $2n$. Let $\xi$ be the Reeb vector field of both cosymplectic structures. $i_\xi\Omega_t = (1-t)i_\xi\Omega_0 + ti_\xi\Omega_1 = 0$ and $\eta(\xi) = 1$ by definition, so $(\Omega_t, \eta)$ is a cosymplectic structure over a neighbourhood of $M$ in $P$. We can use the relative Poincar\'e's lemma (lemma \ref{relative_poincare}) to get a 1-form $\phi$ such that $\omega = d\phi$ and $\phi|_M = 0$.

    We are interested in picking $\phi$ such that $\phi(\xi) = 0$. In order to do that, we will slightly change the 1-form $\phi$ given by the Poincar\'e Lemma \ref{relative_poincare}. Let $\alpha$ be a diffeomorphism from a tubular neighbourhood of $M$ into a neighbourhood $V$ of the zero-section of a vector bundle over $M$. Let $\tilde{\xi} = (d\alpha) \xi$. We define $\tilde{\pi}_t: V \rightarrow V$ as the multiplication by $t$ in all directions except the one given by $\tilde{\xi}$, for $t\in [0,1]$. Note that $\tilde{\pi}_t$ is the same as in the proof of the Poincaré Lemma except for a slight change. We similarly define $X_t$ as ($\frac{d\tilde{\pi}_s}{ds})_{s=t}$.

    The image of $\tilde{\pi}_0$ is a 1-dimensional vector bundle over $M$ whose fibres are in the directions defined by $\tilde{\xi}$. In this bundle, $(d\alpha^{-1})^*\omega = 0$, as it is zero on $M$ and also cancels on the directions defined by $\tilde{\xi}$  ($\; i_{\tilde{\xi}} (d\alpha^{-1})^*\omega = i_\xi \omega \circ (d\alpha^{-1}) = 0)$. Then, we have that $\tilde{\pi}_0^*(d\alpha^{-1})^*\omega = 0$. We are in the same conditions than in the proof of the relative Poincaré lemma, so if we define
    $$\phi = (\alpha^*)\int_0^1 \pi_t^*(\flat_{(d\alpha^{-1})\omega}(X_t)) dt,$$
    then, $\omega = d\phi$ and $\phi|_M = 0$.
    $$\phi(\xi) = \Big(\int_0^1 \tilde{\pi}_t^*(\flat_{(d\alpha^{-1})\omega}(X_t)) dt \Big)(\tilde{\xi}) = \int_0^1 \tilde{\pi}_t^*(\flat_{(d\alpha^{-1})\omega}(X_t) (\tilde{\xi})) = 0,$$
    since, by definition, $\tilde{\pi}_t (\tilde{\xi}) = \tilde{\xi}$ and $\omega(\xi) = 0$.
    

    
    Let $Y_t = \flat_{(\Omega_t, \eta)}^{-1}(- \phi)$, i.e. $i_{Y_t}\Omega_t + \eta(Y_t)\eta = \phi$. If we evaluate on $\xi$, we get that $0 = \phi(\xi) = \eta(Y_t)$. 
    Since $\phi |_M = 0$, then $Y_t|_M = 0$ and the time-dependent vector field $\{Y_t\}_{t \in [0,1]}$ is integrable to a one-parameter family $\{f_t\}_{t \in [0,1]}$ of diffeomorphisms of neighbourhoods of $M$ in $P$ by lemma \ref{lemma:integrable}.

    Then, we have that

    \begin{align*}
        \frac{d}{ds}(f_s^*\Omega_s)_{s=t} &= f_t^*\Big( \frac{d\Omega_s}{ds} \Big)_{s=t} + f_t^*(d(i_{Y_t}\Omega_t)) \\
        &= f_t^*(\omega + d(-\phi)) \\
        &= 0.
    \end{align*}

    Let $f = f_1$. Then $f^*\Omega_1 = f_1^*\Omega_1 = f_0^*\Omega_0 = \Omega_0$. In the same way,

     \begin{equation*}
        \frac{d}{ds}(f_s^*\eta)_{s=t} = f_t^*\Big( \frac{d\eta}{ds} \Big)_{s=t} + f_t^*(d(\eta(Y_t))) = 0.
    \end{equation*}

    So, $f^*\eta = \eta$, and $f$ is a cosymplectomorphism. As $Y_t = 0$ on $M$, the flow in $M$ is the identity, and we have that $f|_M$ and $Tf|_M$ are the identity mappings.
\end{proof}
        


\begin{theorem}
    Let $(\Omega_0, \eta_0)$ and $(\Omega_1, \eta_1)$ cosymplectic structures on $P$ whose restrictions to $T_MP$ are equal. Then there are neighbourhoods $U_1$ and $U_2$ of $M$ in $P$ and a cosymplectomorphism $\phi$ from $(U_0, \Omega_0, \eta_0)$ to $(U_1, \Omega_1, \eta_1)$ such that $\phi|_M$ and $T\phi|_{T_MP}$ are the identity mappings
\end{theorem}

\begin{proof}
    It follows combining the two previous lemmas. From lemma \ref{lemma:cosympl:ext:1} we get a cosymplectomorphsim $g$ between two neighbourhoods $U_0$, $U_1$ of $M$ with cosymplectic structures $(\Omega_0, \eta_0)$ and $(\bar{\Omega}_0, \eta_1)$. Then, we apply lemma \ref{lemma:cosympl:ext:2} and get a cosymplectomorphism $f$ from a neighbourhood $V_1 \subset U_1$ to a neighbourhood $V_2$ with cosymplectic structures $(\bar{\Omega}_0, \eta_1)$ and $(\Omega_1, \eta_1)$. Then, $f \circ (g|_{g^{-1}(V_1)})$ is a cosymplectomorphism from $(g^{-1}(V_1), \Omega_0, \eta_0)$ to $(V_2, \Omega_1, \eta_1)$. 
\end{proof}

We have now completed the proof of the extension of Theorem \ref{sympl:extension:same}. The key idea was to first approach the 1-forms, as trying to estimate a cosymplectic structure for $\Omega_t$ without knowing $\xi_t$ is much more complicated.

The statement for the extension of Theorem \ref{thm:sympl:ext} is the follwing:

\begin{theorem}
    \label{thm:cosympl:ext}
    Let $P$ be a manifold and let $M$ be a closed submanifold of $P$. Let $(\Omega_0, \eta_0)$ and $(\Omega_1, \eta_1)$ be cosymplectic structures on $P_0$ and $P_1$ respectively, and let $M$ be embedded in both $P_0$ and $P_1$. Then, if $(T_MP_0, \Omega_0, \eta_0)$ and $(T_MP_1, \Omega_1, \eta_1)$ are isomorphic as symplectic vector bundles, there are neighbourhoods $U_0, U_1$ of $M$ in $P_0, P_1$ respectively and a cosymplectomorphism $\psi: (U_0, \Omega_0, \eta_0) \rightarrow (U_1, \Omega_1, \eta_1)$ such that $\psi |_M$ is the identity. 
\end{theorem}

The proof of this result is equivalent to the proof of the symplectic case, given in Theorem \ref{thm:sympl:ext} (Appendix A). We now have all the preliminary tools needed to prove the local uniqueness theorem.

\begin{theorem}
    \textbf{(Local Uniqueness) }All coisotropic embeddings of $(M, \Omega, \eta)$ are neighbourhood equivalent.
\end{theorem}

\begin{proof}
    Let $(M, \Omega, \eta)$ be coisotropically embedded in $(P, \bar{\Omega}, \bar{\eta})$. $K = TM^\perp \subset TM$. We define $G$ to be a compliment of $K$ in the same way as in Theorem \ref{thm:cosympl_cois:existence}, i.e. $TM = K \oplus G$ and $\xi_x \in G_x \;\forall x \in M$.

    We have $K = TM^\perp = (K \oplus G)^\perp = K^\perp \cap G^\perp = (TM \; \cap \; \text{ker } \bar{\eta})  \cap G^\perp$. By definition of cosymplectic complement, $G^\perp \subset \text{ ker } \bar{\eta}$, so $G^\perp = G^\perp \cap \text{ ker } \bar{\eta}$. Then, $G \cap G^\perp = G^\perp \cap \text{ ker } \bar{\eta} \cap TM \cap G = K \cap G = \emptyset$. Since
    $(G \oplus G^\perp)^\perp = G^\perp \cap G \cap \text{ ker } \bar{\eta} = \{0\}$, then $T_MP = G \oplus G^\perp$.

    Now, $K \subset G^\perp$ and, since $K^\perp \cap G^\perp = K$, then $K$ is a Lagrangian subbundle of $G^\perp$. Then, by proposition \ref{prop:Lagrangian:symplectomorphism}, there exists a symplectomorphism  $\varphi: (G^\perp, \bar{\Omega} \mid_{G^\perp}) \rightarrow (K \oplus K^*, \omega_K)$, such that $\varphi |_K = Id$, where $\omega_K$ is the canonical sympletic form in $K \oplus K^*$. Then, there is an induced diffeormorphism $\psi: G \oplus G^\perp \rightarrow G \oplus K \oplus K^*$ given by $\psi(g_1, g_2) = (g_1, \varphi(g_2))$, for $g_1 \in G, g_2 \in G^\perp$.

    Let $\Omega_G$ be the pullback of $\bar{\Omega}$ by $\psi^{-1}$. We have $\psi^{-1}(G \oplus K) = TM$ since $\psi^{-1} |_G$ and $\varphi^{-1} |_K$ are the identity and $K \oplus G = TM$. So, $\Omega_G |_{G \oplus K} = \Omega |_{TM}$, by definition of coisotropic embedding. Thus, $\psi^{-1}(K \oplus K) = G^\perp$. Since $\varphi$ is a symplectomorphism, then we have that $\Omega_G |_{K \oplus K^*} = (\varphi^{-1})^*\Omega|_{G^\perp} = \omega_K$. Also, using that $\Omega_G$ is defined as the pullback of a diffeomorphism, we have that in $G \oplus K \oplus K^*$, $G^\perp = K \oplus K^*$. We can decompose any element $X$ of $G \oplus K \oplus K^*$ as $X_1 + X_2$, where $X_1 \in G, X_2 \in K \oplus K^*$. Then, by linearity, 
    \begin{align*}
        \Omega_G(X, Y) &= \Omega_G(X_1 + X_2, Y_1 + Y_2) \\
        &= \Omega_G(X_1, Y_1) + \Omega_G(X_1, Y_2) + \Omega_G(X_2, Y_1) + \Omega_G(X_2, Y_2) \\
        &= \Omega(X_1, Y_1) + \omega_K(X_2, Y_2),
    \end{align*}

\noindent   where $\Omega_G(X_1, Y_2) = 0$ and $\Omega_G(X_2, Y_1) = 0$ because one of the vectors is in $G$ and the other in $G^\perp$. So, $\Omega_G = \pi^*\Omega + j_G^*\omega_K$. Here, $\pi$ and $j_G$ are the same mappings as the ones defined in the existence theorem \ref{thm:cosympl_cois:existence}. Note also that $j_G^*\omega_K = -\omega_{p_G}$ (as defined by \cite{marle83sous}).

    Let $\eta_G$ be the pullback of $\bar{\eta}$ by $\psi^{-1}$. $\eta_G |_{K \oplus K^*} = 0$ and $\eta_G |_{G} = \eta$. Consequently, the cosymplectic vector bundle $(G \oplus K \oplus K^*, \Omega_G, \eta_G)$ depends only upon $M, \Omega, \eta$ and the decomposition $K \oplus G$ of $TM$ and, thus, it is independent of the embedding space $(P, \bar{\Omega}, \bar{\eta})$. If $(P', \Omega', \eta')$ is another cosymplectic manifold in which $(M, \Omega, \eta)$ is coisotropically embedded, then, the previous constructions provide vector bundle cosymplectomorphisms 
    $$
    (T_MP, \bar{\Omega}, \bar{\eta}) \sim (G \oplus K \oplus K^*, \Omega_G, \eta_G) \sim (T_MP', \Omega', \eta')
    $$ 
    for some fixed (but irrelevant) splitting.

    Next, theorem \ref{thm:cosympl:ext} implies that there exist neighbourhoods of $M$ $U_0, U_1$, in $P$ and $P'$ respectively, and a cosymplectomorphism $\psi: (U_0, \bar{\Omega}, \bar{\eta}) \rightarrow (U_1, \Omega', \eta')$ such that $\psi|_M$ is the identity. So, $(P_0, \bar{\Omega}, \bar{\eta})$ and $(P_1, \Omega', \eta')$ are neighbourhood equivalent.
\end{proof}

\section{Conclusions and future work}

In this paper we have extended the so-called coisotropic embedding theorem for presymplectic manifolds to the context of precosymplectic manifolds. The reason for this extension lies in the need to have a natural geometrical framework for time-dependent Hamiltonian or Lagrangian systems. Although this extension may seem simple, it has nevertheless required devising new uses of key results such as the Relative Poincaré Lemma and various results of Alan Weinstein.

We propose the following challenges for future work:

\begin{enumerate}

\item To apply the results obtained in this current work to develop a regularization method similar to that developed by Alberto Ibort and Jes\'us Mar{\'\i}n-Solano \cite{ibort-marin}.

\item In recent years, there has been a lot of attention to contact Hamiltonian and Lagrangian systems (see, for instance, \cite{contact}). We would like to obtain a similar coisotropic embedding Theorem in this case. The relevance of contact systems is that they experiment dissipation (in the Lagrangian description, these systems are referred as Lagrangians depending on the action). They are very relevant in many fields, among them, thermodynamics.

\end{enumerate}

\section*{Acknowledgments}

We acknowledge financial support of the 
{\sl Ministerio de Ciencia, Innovaci\'on y Universidades} (Spain), grants PID2021-125515NB-C21 and RED2022-134301-T.
We also acknowledge financial support from the Severo Ochoa Programme for Centers of Excellence in R\&D (CEX2019-000904-S). Pablo Soto would like to thank the JAE Intro program of the CSIC for providing me with an introductory research grant at the Institute of Mathematical Sciences (ICMAT).

\bibliographystyle{plain}
\bibliography{sample}

\begin{thebibliography}{10}

\bibitem{Abraham1967Foundations}
R.~Abraham and J.~E. Marsden.
\newblock {\em {Foundations of Mechanics}}.
\newblock Mathematical Physics Monograph Series. W. A. Benjamin, 1967.

\bibitem{albert}
C.~Albert.
\newblock Le théorème de réduction de marsden-weinstein en géométrie cosymplectique et de contact.
\newblock {\em J. Geom. Phys.}, 6(4):627–649, 1989.

\bibitem{cantrijn}
F.~Cantrijn, M.~de~León, and E.~A. Lacomba.
\newblock Gradient vector fields on cosymplectic manifolds.
\newblock {\em J. Phys. A}, 25(1):175–188, 1992.

\bibitem{CAPPELLETTI_MONTANO_2013}
B.~Cappelletti-Montano, A.~De~Nicola, and I.~Yudin.
\newblock A survey on cosymplectic geometry.
\newblock {\em Reviews in Mathematical Physics}, 25(10):1343002, November 2013.

\bibitem{chinea}
D.~Chinea, M.~de~León, and J.~C. Marrero.
\newblock The constraint algorithm for time-dependent lagrangians.
\newblock {\em J. Math. Phys.}, 35(7):3410–3447, 1994.

\bibitem{deLeon2011methods}
M.~de~Le{\'o}n and P.R. Rodrigues.
\newblock {\em Methods of Differential Geometry in Analytical Mechanics}.
\newblock ISSN. Elsevier Science, 1989.

\bibitem{cosimpl}
M.~de~León and R.~Izquierdo-López.
\newblock A review on coisotropic reduction in symplectic, cosymplectic, contact and co-contact hamiltonian systems.
\newblock {\em J. Phys. A: Math. Theor.}, 57(163001 (50pp)), 2024.

\bibitem{contact}
M.~de~León and M.~Lainz~Valcázar.
\newblock Contact hamiltonian systems.
\newblock {\em J. Math. Phys.}, 60(10):102902, 18 pp., 2019.

\bibitem{Dirac_1950}
P.A.M. Dirac.
\newblock Generalized {H}amiltonian dynamics.
\newblock {\em Canadian Journal of Mathematics}, 2:129–148, 1950.

\bibitem{godbillon}
C.~Godbillon.
\newblock {\em G{\'e}om{\'e}trie diff{\'e}rentielle et m{\'e}canique analytique}.
\newblock Editions Hermann, Paris, 1969.

\bibitem{gotay1982coisotropic}
M.~J. Gotay.
\newblock On coisotropic embeddings of presymplectic manifolds.
\newblock {\em Proceedings of the American Mathematical Society}, 84:111--114, 1982.

\bibitem{gotay1}
Mark~J. Gotay, James~M. Nester, and George Hinds.
\newblock Presymplectic manifolds and the dirac-bergmann theory of constraints.
\newblock {\em J. Math. Phys.}, 19(11):2388–2399, 1978.

\bibitem{GotaySingularLagrangians1}
M.J. Gotay and J.M. Nester.
\newblock Presymplectic {L}agrangian systems. {I}. {T}he constraint algorithm and the equivalence theorem.
\newblock {\em Ann. Inst. H. Poincar\'{e} Sect. A (N.S.)}, 30(2):129--142, 1979.

\bibitem{GotaySingularLagrangians2}
M.J. Gotay and J.M. Nester.
\newblock Presymplectic {L}agrangian systems. {II}. {T}he second-order equation problem.
\newblock {\em Ann. Inst. H. Poincar\'{e} Sect. A (N.S.)}, 32(1):1--13, 1980.

\bibitem{ibort-marin}
A~Ibort and J.~Marín-Solano.
\newblock Coisotropic regularization of singular lagrangians.
\newblock {\em J. Math. Phys.}, 36(10):5522–5539, 1995.

\bibitem{libermann1959sur}
P.~Libermann.
\newblock Sur les automorphismes infinit{\'e}simaux des structures symplectiques et des structures de contact.
\newblock In {\em Colloque G{\'e}om. Diff. Globale (Bruxelles, 1958)}, pages 37--59. Centre Belge Rech. Math., Louvain, 1959.

\bibitem{marle83sous}
Ch-M. Marle.
\newblock Sous-variétés de rang constant d’une variété symplectique.
\newblock {\em Astérisque.}, 107, 01 1983.

\bibitem{zambon2014equivalences}
F.~Schaetz and M.~Zambon.
\newblock Equivalences of coisotropic submanifolds.
\newblock {\em Journal of Symplectic Geometry}, 15, 11 2014.

\bibitem{weinstein71symplectic}
A.~Weinstein.
\newblock Symplectic manifolds and their lagrangian submanifolds.
\newblock {\em Advances in Mathematics}, 6(3):329--346, 1971.

\bibitem{weinstein1977lectures}
A.~Weinstein and Conference~Board of~the Mathematical~Sciences.
\newblock {\em Lectures on Symplectic Manifolds}.
\newblock Regional conference series in mathematics. Conference Board of the Mathematical Sciences, 1977.

\end{thebibliography}

\appendix

\section{Coisotropic embeddings of presymplectic manifolds}

In this section, we present different results regarding coisotropic submanifolds of symplectic manifolds, based on the work of Gotay \cite{gotay1982coisotropic}, Schaetz and Zambon \cite{zambon2014equivalences} and Marle \cite{marle83sous}. These statements have been generalised to cosymplectic manifolds in this paper.

The main result is the following: 

\begin{theorem}
    \label{thm:coisotropic:global}
    Every coisotropic submanifold has an induced pre-symplectic structure. Additionally, every pre-symplectic manifold can be coisotropically embedded into a symplectic manifold. This embedding is unique up to a local symplectomorphism.
\end{theorem}

Theorem \ref{thm:coisotropic:global} can be divided into three smaller results, one corresponding to each statement of the theorem. We will denote these results as direct, existence and local uniqueness theorems respectively. We will provide proofs of these results for the sake of completeness.

\subsection{Direct theorem}

The direct statement can be found in \cite{zambon2014equivalences}. We present an alternative proof.

\begin{theorem}
    \textbf{(Direct Theorem)}
    Let $(P, \omega)$ be a symplectic manifold, and let $M$ be a coisotropic submanifold. The pull-back of $\omega$ to $M$ along the inclusion $j: M \xhookrightarrow{} P$ is a closed two-form of constant rank $2\textrm{dim} M - \textrm{dim}P$.
\end{theorem}

\begin{proof}
Let $\omega_M = j^*\omega$. $\omega_M$ is a closed two-form as $d(j^*\omega) = j^*d\omega = 0$. Let $x$ be a point of $M$. As $P$ is symplectic, we have that ${\text{dim }T_xP = \text{dim }T_xM + \text{dim }(T_xM)^\perp}$. So, we have that ${\text{dim }(T_xM)^\perp = \text{dim }T_xP - \text{dim }T_xM}$. As $M$ is a coisotropic submanifold, we have that ${(T_xM)^\perp \subset T_xM}$, so ${(T_xM)^\perp = \text{ker }\flat_{\omega_M (x)}}$. Then,
\begin{flalign*}
    \text{rank }\omega_M(x) &= \text{dim (Im }\flat_{\omega_M (x)}) = \text{dim }T_xM - \text{dim (ker }\flat_{\omega_M(x)}) = 2 \text{dim }T_xM - \text{dim }T_xP \\
    &= 2 \text{dim }M - \text{dim }P. &&
\end{flalign*}
\end{proof}

\subsection{Existence theorem}

In 1982, Gotay \cite{gotay1982coisotropic} proved existence and uniqueness theorems for coisotropic embeddings of presymplectic manifolds. His work provides a complete local characterisation of coisotropic embeddings of presymplectic manifolds into symplectic manifolds. The proof below is also based on the ideas of Marle \cite{marle83sous}.

In order to prove the existence theorem, we will embbed $M$ into the zero-section of the dual of the characteristic bundle. The characteristic bundle $K$ is a vector bundle defined by the fibres $$K_x = \{X \in T_xM, \; i_X\omega = 0\}.$$ Then, $K^*$ is a vector bundle whose fibres are $$(K^*)_x = (K_x)^*,$$ the dual of the fibres of $K$.

\begin{theorem}
    \label{thm:sympl_cois:existence}
    \textbf{(Existence Theorem)} Let $(M, \omega)$ be a presymplectic manifold of dimension $2p + k$, where $\omega$ has constant rank $2p$. There exists a symplectic structure on a neighbourhood of the zero-section of $K^*$, where $K$ denotes the characteristic bundle of $(M, \omega)$. Moreover, $(M, \omega)$ may be coisotropically embedded in this neighbourhood as the zero-section.

\end{theorem}

\begin{proof}
    Let $K$ be the kernel of $\omega$. Let $G$ be a complement of $K$. Then, $TM = K \oplus G$. Thus, $T_MK^*$ has the canonical decomposition $T_MK^* = TM \oplus K^* = G \oplus K \oplus K^*$. Let $p_G: TM \rightarrow K$ be the projection such that for $v = (v_k, v_g) \in TM$, with $v_k \in K$ and $v_g\in G$, then $p_G(v) = p_G((v_k, v_g)) = v_k$. Let $j_G$ be the adjoint of $p_G$, i.e. if $\alpha \in K^*, \; j_G(\alpha) = \alpha \circ p_G$.
    
    As $\omega$ is presymplectic of constant rank $2p$, it is possible to find local Darboux coordinates in $M$, $(x^i, z^\mu), \; i=1, ..., 2p,\; \mu=1, ..., k$ such that $\omega = \sum_{i=1}^p dx^i \wedge dx^{p+i}$. The coordinates $z^\mu$ are local coordinates along the fibres of $K$, as the vector fields $\partial/\partial z^\mu$ form a local basis of $\Gamma(K)$. So, for any element of $K^*$, there exist local coordinates $(x^i, z^\mu, b_\mu)$, such that $b_r(\frac{\partial}{\partial z^s}) = \delta_{r, s}$. Then, $\pi^*\omega = \sum_{i=1}^p dx^i\wedge dx^{p + i}$, where $\pi$ is the canonical projection $\pi: K^* \rightarrow M$.




    Any vector $X \in TK^*$ can be expressed in the form $$X = \sum_{i=1}^{2p} X^i \frac{\partial}{\partial x^i} + \sum_{r=1}^k Z^r \frac{\partial}{\partial z^r} + \sum_{s=1}^k B_s \frac{\partial}{\partial b_s},$$ where $X^i, Z^r, B_s$ are functions of the local coordinates $(x^i, z^\mu, b_\mu)$. We have that

    \begin{align*}
        T\pi(X) &= \sum_{i=1}^{2p} X^i \frac{\partial}{\partial x^i} + \sum_{r=1}^k Z^r \frac{\partial}{\partial z^r},\\
        p_G \circ T\pi(X) &= \sum_{r=1}^k \Big(zZ^r + \sum_{i=1}^{2p} A^r_i X^i \Big)\frac{\partial}{\partial z^r},
    \end{align*}

    \noindent where the $A^r_i$ depend on $(x^i, z^\mu)$ and the projection $p_G$. Let $\lambda_{p_G}$ be the generalised Liouville form \cite{marle83sous}. Then, $$\lambda_{p_G}(X) = \sum_{r=1}^k B_r\Big(Z^r + \sum_{i=1}^{2p} A^r_i X^i\Big).$$ 
    We can express the generalised Liouville 1-form as 
    $$
    \lambda_{p_G} = \sum_{r=1}^k b_r \Big(dz^r + \sum_{i=1}^{2p} A^r_i dx^i\Big).
    $$

    Let $\Omega_G = \pi^*\omega - \omega_{p_G}$ a 2-form on $K^*$. Then, $\Omega_G$ is locally given by

    \begin{align}
        \label{sympl_cois:existence:omega}
        \Omega_G =& \sum_{i=1}^p dx^i\wedge dx^{p + i} + \sum_{r = 1}^k db_r \wedge \Big(dz^r + \sum_{i=1}^{2p} A^r_i dx^i\Big) \\
        &+ \sum_{r = 1}^k b_r\Big[ \sum_{i=1}^{2p} \Big( \sum_{j=1}^{2p} \frac{\partial A^r_i}{\partial x^j} dx^j +\sum_{s=1}^k \frac{\partial A^r_i}{\partial z^s} dz^s \Big) \wedge dx^i\Big] \nonumber
    \end{align}
    
    $\Omega_G$ is smooth since it is a composition of smooth forms and it is well-defined. Furthermore, $\Omega_G$ is closed:
    $$d\Omega_G = d(\pi^*\omega + d\lambda_{p_G}) = \pi^*(d\omega) + dd\lambda_{p_G} = 0.$$ 
    Along the zero section of $K^*$, $b_r = 0$ and we can see that $\Omega_G$ is non-degenerate. So, there exists an neighbourhood of the zero section of $K^*$ such that $\Omega_G$ is non-degenerate. 
    


    To see that $M$ is a coisotropic submanifold of $K^*$, we need to show that $TM^\perp \subset TM$ in $T_MK^*$. Assume that this is not the case. Then, there exists $x \in M$, $Y \in T_xK^*, Y \not\in T_xM$, such that $\Omega_G(X,Y) = 0, \forall X \in TM$. If we use Darboux local coordinates, then we get that there is some $b_\mu \neq 0$. Then, if we take $X = \partial/\partial z^\mu$, we would have $\Omega_G(X,Y) = 1$. 

    

\end{proof}

\subsection{Local uniqueness theorem}

In order to prove the local uniqueness theorem, some previous results from Weinstein \cite{weinstein71symplectic}, \cite{weinstein1977lectures} are needed. It is important to study these results with detail, as they have some key ideas that will help us generalise the results to cosymplectic manifolds. We will make use of a generalisation of the Poincaré lemma.

\begin{lem}
    \label{relative_poincare}
    \textbf{(Relative Poincaré Lemma)} Let $P$ be a manifold and $M$ a closed submanifold of $P$. Let $\Omega$ be a closed $k$-form on $P$ whose pullback to $M$ is zero. Then there is a $k-1$-form $\phi$ on a neighbourhood of $M$ in $P$ such that $d\phi = \Omega$ and $\phi$ vanishes on $M$. If $\Omega$ vanishes on $N$, then $\phi$ can be chosen so that its first partial derivatices vanish on $M$.   
\end{lem}

\begin{proof}
    We include an sketch of the proof. The full version is available in \cite{weinstein71symplectic}. 
    
    The proof of the relative Poincaré lemma uses the homotopy operator associated with a deformation of a tubular neighbourhood of $M$. We may assume that the neighbourhood $U$ of $P$ is a vector bundle whose zero-section is $M$. We define $\pi_t: U \rightarrow U$ as the multiplication by $t$, for $t \in [0, 1]$. We define the vector field $X_t$ as ($\frac{d\pi_s}{ds})_{s=t}$. Then, if $\Omega$ is a closed p-form on $U$ which is zero on $M$, we deduce
    
    \begin{equation}
        \frac{d}{ds}(\pi_s^*\Omega)_{s=t} = \pi_t^*(\flat_\Omega(X_t)) + d[\pi_t^*(\flat_\Omega(X_t)] = d[\pi_t^*(\flat_\Omega(X_t)].
    \end{equation}
    
    If we integrate with respect to $t$ on $[0,1]$ and use that $\pi_0^*\Omega = 0$ and $\pi_1$ is the identity, we get that
    
    \begin{equation}
        \Omega = \int_0^1 d(\pi_t^*(\flat_\Omega(X_t))) dt.
    \end{equation}
    
    If we define $\phi = \int_0^1 \pi_t^*(\flat_\Omega(X_t)) dt$, then $\Omega = d\phi$.
    
\end{proof}

To prove the local uniqueness theorem, we need to build symplectomorphisms between tubular neighbourhoods. In order to construct these diffeomorphisms, we will make use of the flow defined by vector fields. A time-dependent vector field $\{Y_t\}_{t\in I}$ is \textit{integrable} if there exists a one-parameter family $\{f_t\}_{t\in I}$ of diffeomorphisms such that $f_0$ is the identity and $\frac{df_t}{dt} = Y_t \circ f_t$. Then,

\begin{lem}
    \label{lemma:integrable}
    Let $P$ be a manifold and $M$ a closed submanifold of $P$. Let $\{Y_t\}_{t\in I}$ be a time-dependent vector field on a neighbourhood of $M$ in $P$. If $Y_t|_M = 0$, then $\{Y_t\}_{t\in I}$ is integrable.
\end{lem}

The proof is based on the openness of the domain of definition of the solutions of a system of ordinary differential equations. The following theorem, which is a generalisation of the Darboux theorem, will enable us to find symplectomorphisms over symplectic structures on the same manifold \cite{weinstein71symplectic}:

\begin{theorem}
    \label{sympl:extension:same}
    Let $P$ be a manifold and $M$ a closed submanifold. Let $\Omega_0$ and $\Omega_1$ be symplectic structures over a neighbourhood $U$ of $M$ in $P$ such that $\Omega_0 |_M = \Omega_1 |_M$. Then, there exists a symplectomorphism from $f:U \rightarrow U$ such that $f |_M$ is the identity. 
\end{theorem}

\begin{proof}
    Let $\omega = \Omega_1 - \Omega_0$ and let $\Omega_t = \Omega_0 + t\omega$ for $t$ between 0 and 1, and, since $\Omega_0, \Omega_1$ are closed, we have $d\omega = 0$ and $\omega |_M = 0$. In addition, $d\Omega_t = 0$ and $\Omega_t |_M = \Omega_0 |_M$ and $\Omega_t$ is non-degenerate over $M$; so, there is a neighbourhood of $M$ in $P$ such that $\Omega_t$ is symplectic. We can use the relative Poincar\'e's lemma to get a 1-form $\phi$ such that $\omega = d\phi$ and $\phi|_M = 0$. Let $Y_t = \flat_{\Omega_t}^{-1}(- \phi)$. $\flat_{\Omega_t}(X)$ is defined as $i_X\Omega_t$ and it is an isomorphism of $T_xP$ and $T_x^*P$ (\cite{deLeon2011methods}). Since $\phi |_M = 0$, then $Y_t|_M = 0$ and the time-dependent vector field $\{Y_t\}_{t \in [0,1]}$ is integrable to a one-parameter family $\{f_t\}_{t \in [0,1]}$ of diffeomorphisms of neighbourhoods of $M$ in $P$ (the result can be found in Lemma 3.1 of \cite{weinstein71symplectic}). Then, we have that

    \begin{align*}
        \frac{d}{ds}(f_s^*\Omega_s)_{s=t} &= f_t^*\Big( \frac{d\Omega_s}{ds} \Big)_{s=t} + f_t^*(d(i_{Y_t}\Omega_t)) \\
        &= f_t^*(\omega + d(-\phi)) \\
        &= 0.
    \end{align*}

Let $f = f_1$. Then $f^*\Omega_1 = f_1^*\Omega_1 = f_0^*\Omega_0 = \Omega_0$.

\end{proof}

We can extend this theorem to construct symplectomorphisms between tubular neighbourhoods of different manifolds. The result is called Weinstein's extension theorem \cite{weinstein1977lectures}:

\begin{theorem}
    \label{thm:sympl:ext}
    \textbf{(Extension theorem)} Let $P$ be a manifold and let $M$ be a closed submanifold of $P$.
    \begin{itemize}
        \item If $(T_MP, \Omega)$ is a symplectic vector bundle such that $\Omega |_{TM}$ is closed, then $\Omega$ extends to a symplectic structure on a neighbourhood of $M$ in $P$.
        \item Let $\Omega_0$ and $\Omega_1$ be symplectic forms on $P_0$ and $P_1$ respectively, and let $M$ be embedded in both $P_0$ and $P_1$. Then, if $(T_MP_0, \Omega_0)$ and $(T_MP_1, \Omega_1)$ are isomorphic as symplectic vector bundles, there are neighbourhoods $U_0, U_1$ of $M$ in $P_0, P_1$ respectively and a symplectomorphism $\psi: (U_0, \Omega_0) \rightarrow (U_1, \Omega_1)$ such that $\psi |_M$ is the identity. 
    \end{itemize}
\end{theorem}

\begin{proof}
    To prove the local uniqueness theorem we will only use the second part of the extension theorem. 
    
    Let $V_0$, $V_1$ be tubular neighbourhoods of $P_0$ and $P_1$ respectively. Let $\alpha_0, \alpha_1$ be the respective diffeomorphisms to neighbourhoods of the zero section of $T_MP_0$ and $T_MP_1$, and let $\varphi$ be the isomorphism between the two tangent bundles $T_MP_0$ and $T_MP_1$. If $V_1$ is sufficiently small, then $\phi = \alpha_0^{-1} \circ \varphi \circ \alpha_1$ is a diffeomorphism between $V_1$ and a neighbourhood $\tilde{V}_1$ of $M$ in $P_0$. Then, $(\tilde{V}_1, \phi^*\Omega_1)$ is a neighbourhood of $M$ in $P_0$. Let $U_0 = V_0 \cap \tilde{V}_1$. By theorem \ref{sympl:extension:same}, there exists a symplectomorphism $f: (U_0, \Omega_0) \rightarrow (U_0, \phi^*\Omega_1)$ such that $f|_M$ is the identity. Then, $f \circ \phi^{-1}$ is a symplectomorphism between $(U, \Omega_0$) and some neighbourhood $(U_1, \Omega_1)$ of $M$ in $P_1$.
    
\end{proof}

\begin{prop}
    \label{prop:Lagrangian:symplectomorphism}
    Let $(E, \Omega)$ be a symplectic vector bundle over $M$, and let $L \subset E$ be a Lagrangian subbundle. Then, there is a symplectomorphism $\varphi: (E, \Omega) \rightarrow (L \oplus L^*, \omega_L)$, such that $\varphi|_L$ is the identity. $\omega_L$ is the canonical symplectic form on $L \oplus L^*$, i.e. $\omega_L((l \oplus l^*), (m\oplus m^*)) = m^*(l) - l^*(m)$.
\end{prop}

The canonical symplectic form $\omega_L$ is analogous to the canonical form defined by the Liouville form. The proof of this result can be found in \cite{Abraham1967Foundations}.

With these tools, we can now prove the local uniqueness theorem. The proof is based on the work of Gotay \cite{gotay1982coisotropic}.

\begin{theorem}
    \textbf{(Local Uniqueness Theorem)} All coisotropic embeddings of $(M, \omega)$ are neighbourhood equivalent.
\end{theorem}

\begin{proof}
    Let $(M, \omega)$ be coisotropically embedded in $(P, \Omega)$. We define $K = TM^\perp \subset TM$. Let $G$ be a complement of $K$ such that $TM = K \oplus G$. We have

    $K = TM^\perp = (K \oplus G)^\perp = K^\perp \cap G^\perp = TM \cap G^\perp$. So, $G \cap G^\perp = G^\perp \cap (TM \cap G) = K \cap G = \emptyset$. $(G \oplus G^\perp)^\perp = G^\perp \cap G = \emptyset$. Then, $G$ is symplectic and $T_MP = G \oplus G^\perp$.

    Since $K \subset G^\perp$ and, as $K^\perp \cap G^\perp = K$, then $K$ is a Lagrangian subbundle of $G^\perp$. Thus, by proposition \ref{prop:Lagrangian:symplectomorphism}, there exists a symplectomorphism  $\varphi: (G^\perp, \Omega \mid_{G^\perp}) \rightarrow (K \oplus K^*, \omega_K)$, such that $\varphi |_K = Id$. There is an induced diffeormorphism $\psi: G \oplus G^\perp \rightarrow G \oplus K \oplus K^*$ given by $\psi(g_1, g_2) = (g_1, \varphi(g_2))$, for $g_1 \in G, g_2 \in G^\perp$.


    Let $\Omega_G$ be the pullback of $\Omega$ by $\psi^{-1}$. So, $\psi^{-1}(G \oplus K) = TM$ as $\psi^{-1} |_G$ and $\varphi^{-1} |_K$ are the identity and $K \oplus G = TM$. So, $\Omega_G |_{G \oplus K} = \Omega |_{TM} = \omega$, by definition of coisotropic embedding. Also, $\psi^{-1}(K \oplus K^*) = G^\perp$. As $\varphi$ is a symplectomorphism, we have that $\Omega_G |_{K \oplus K^*} = (\varphi^{-1})^*\Omega|_{G^\perp} = \omega_K$. Also, as $\Omega_G$ is defined as the pullback of a diffeomorphism, we have that in $G \oplus K \oplus K^*$, $G^\perp = K \oplus K^*$. We can decompose any element $X$ of $G \oplus K \oplus K$ as $X_1 + X_2$, where $X_1 \in G, X_2 \in K \oplus K^*$. Then, by linearity, 
    \begin{align*}
        \Omega_G(X, Y) &= \Omega_G(X_1 + X_2, Y_1 + Y_2) \\
        &= \Omega_G(X_1, Y_1) + \Omega_G(X_1, Y_2) + \Omega_G(X_2, Y_1) + \Omega_G(X_2, Y_2) \\
        &= \omega(X_1, Y_1) + \omega_K(X_2, Y_2),
    \end{align*}

    \noindent where $\Omega_G(X_1, Y_2) = 0$ and $\Omega_G(X_2, Y_1) = 0$ since one of the vectors is in $G$ and the other one in $G^\perp$. So, $\Omega_G = \pi^*\omega + j_G^*\omega_K$. Here, $\pi$ and $j_G$ are the same mappings as the ones defined in the existence theorem \ref{thm:sympl_cois:existence}. Note also that $j_G^*\omega_K = -\omega_{p_G}$.

    Consequently, the symplectic vector bundle $(G \oplus K \oplus K^*, \Omega_G)$ depends only upon $M, \omega$ and the decomposition $K \oplus G$ of $TM$ and, thus, it is independent of the embedding space $(P, \Omega)$. If $(P', \Omega')$ is another symplectic manifold in which $(M, \omega)$ is coisotropically embedded, then, the previous constructions provide a vector bundle symplectomorphisms $$(T_MP, \Omega) \sim (G \oplus K \oplus K^*, \Omega_G) \sim (T_MP', \Omega')$$ for some fixed (but irrelevant) splitting. We can now apply the extension theorem (\ref{thm:sympl:ext}) to see that $(P, \Omega)$ and $(P', \Omega')$ are neighbourhood equivalent.

\end{proof}

\end{document}